\documentclass[12pt]{article}

\usepackage{amsmath} 
\usepackage{amssymb}
\usepackage{amstext}
\usepackage{amsthm}
\usepackage{amsfonts}
\usepackage{mathtools}
\usepackage{enumerate}  
\usepackage{dsfont}
\usepackage[utf8]{inputenc}

\usepackage{graphicx} 

\theoremstyle{plain}

\newtheorem{theorem}{Theorem}
\newtheorem*{theorem*}{Theorem}
\newtheorem{lemma}[theorem]{Lemma}

\usepackage{notation}

\begin{document}

\renewcommand{\labelenumi}{\roman{enumi}.}

\title{Pretty good state transfer between internal nodes of paths}

\author{  Gabriel Coutinho\thanks{ Dept.\ Computer Science,  IME - University of São Paulo, \texttt{coutinho@ime.usp.br}} \and Krystal Guo\thanks{Dept.\ of Combinatorics and Optimization, University of Waterloo, \newline \texttt{ \{kguo, cvanbomm\}@uwaterloo.ca}} \and Christopher M.~van Bommel\footnotemark[2]}

\date{  \today  }
\maketitle
\begin{abstract}
In this paper, we show that, for any odd prime $p$ and positive integer $t$, the path on  $2^t p -1$ vertices admits pretty good state transfer between vertices  $a$ and $(n+1-a)$ for each $a$ that is a multiple of $2^{t-1}$ with respect to the quantum walk model determined by the XY-Hamiltonian. This gives the first examples of pretty good state transfer occurring between internal vertices on a path, when it does not occur on the extremal vertices.

\vspace{5pt} \noindent Keywords: quantum walks, state transfer, graph eigenvalues 

\vspace{5pt} \noindent Mathematical Subject Classification: 81P68, 15A16, 05C50
\end{abstract}

\section{Introduction and preliminary definitions}

Many quantum algorithms may be modelled as a quantum process occurring on a graph. In \cite{ChildsUniversalQComputation}, Childs shows that any quantum computation can be encoded as a quantum walk in some graph and thus quantum walks can be regarded as a quantum computation primitive. 

We model a network of $n$ interacting qubits by a simple graph $G$ on $n$ vertices, where vertices correspond to qubits and edges to interactions. The interactions are defined by a time-independent Hamiltonian; that is, a symmetric matrix that acts on the Hilbert space of dimension $2^n$. In this paper, we are concerned with the $XY$-Hamiltonian, whose action on the $1$-excitation subspaces is equivalent to the action of the $01$-symmetric adjacency matrix of $G$ on $\C^n$.

Assume all qubits are initialized at the state $|0\rangle$, except for the qubit corresponding to vertex $a$, which is initialized at state $| 1 \rangle$. Let $A = A(G)$ be the adjacency matrix of $G$. According to the Schrödinger equation,
\[|\langle a |\exp(\ii t A) | b \rangle |^2\]
indicates the probability that the state $|1\rangle$ is measured at $b$ after time $t$. Whenever there is a time $t$ such that this probability is equal to $1$, we say that there is \textsl{perfect state transfer} between $a$ and $b$. Perfect state transfer has been studied in many families of graphs, including circulants \cite{BasicCirculant, BasicPetkovicPSTCirculant}, cubelike graphs \cite{GodsilBernasconiPSTCubelike, GodsilCheungPSTCubelike} and distance regular graphs \cite{CoutinhoGodsilGuoVanhove2}. See \cite{GodsilStateTransfer12} for a survey of recent results. 

For paths (or linear chains) with $n$ vertices, perfect state transfer occurs if and only if $n$ is $2$ or $3$. In general, perfect state transfer is understood to be a rare phenomenon, which motivates the definition of \textsl{pretty good state transfer} (or \textsl{almost state transfer}), which is said to occur if, for any $\epsilon > 0$, there is a time $t$ such that
\[|\langle a |\exp(\ii t A) | b \rangle | > 1 - \epsilon.\]
Pretty good state transfer was first studied for paths in \cite{GodsilKirklandSeveriniSmithPGST} and \cite{VinetZhedanovAlmost}. Proofs and methods in this area are a mix of interesting applications of number theory and algebraic graph theory. 

Pretty good state transfer between the end vertices in paths with $n$ elements occurs if and only if $n+1$ is prime, twice a prime or a power of $2$ (see \cite{GodsilKirklandSeveriniSmithPGST}). Moreover, in these cases, it occurs between any pair of vertices equally distant from the centre of the path. In contrast, the question of whether pretty good state transfer is possible between internal vertices in paths when it does not occur between the end vertices gives rise to a more interesting story and does not appear to have a simple answer. This question was raised in \cite{BanchiCoutinhoGodsilSeverini}.

In this paper, we exhibit an infinite family of paths that admit pretty good state transfer between inner vertices but not between the two end vertices. These paths have $2^t p-1$ vertices where $p$ is an odd prime and $t > 1$. We will make use of the following definitions and preliminary results.

Given a symmetric matrix $M$ with $d$ distinct eigenvalues $\theta_1 > ... > \theta_d$, we may write the spectral decomposition of $M$ as follows:
\[M = \sum_{j = 1}^d \theta_j E_j,\]
where $E_r$ denotes the idempotent projection onto the eigenspace corresponding to $\theta_r$. Note that $E_rE_s= \delta_{r,s} E_r$. 

Given a vertex $a \in V(G) = \{1,...,n\}$, let $| a \rangle$ denote the $01$-vector that is $1$ at the entry corresponding to $a$, and $0$ elsewhere. We define the \textsl{eigenvalue support} of $a$ as a subset of the distinct eigenvalues as follows:
\[\Theta_a = \{ \theta_j : E_j |a \rangle \neq 0 \}.\]
We say that vertices $a$ and $b$ are \textsl{strongly cospectral} if $E_j\ket{a} = \pm E_j \ket{b} $ for all idempotents $E_j$ in the spectral decomposition. This property is necessary for both perfect and pretty good state transfer (see \cite[Lemma 3]{BanchiCoutinhoGodsilSeverini}).

The spectrum of the adjacency matrix of the path (see \cite{BrouwerHaemers} for example) on $n$ vertices is 
\begin{align}\theta_j = 2\cos \frac{\pi j}{n+1} \quad \text{ for } j = 1, \ldots, n. \label{eq:1} \end{align}
Note that they are indexed such that $\theta_1 > \ldots > \theta_n$. For each $j$, the eigenvector corresponding to $\theta_j$ is given by
$(\beta_1, \ldots , \beta_n)$ where  $\beta_k = \sin(k\pi j / (n+1))$. The following lemma immediately follows.

\begin{lemma}\label{eq:3} Vertices $a$ and $b$ of $P_n$ are strongly cospectral if and only if $a+b = n+1$.  \end{lemma}

We will make use of the following result, which is derived from \cite[Theorem 2]{BanchiCoutinhoGodsilSeverini}, and completely characterizes pretty good state transfer. The second condition is an immediate consequence of Kronecker's theorem (see \cite[Theorem 4]{BanchiCoutinhoGodsilSeverini}).

\begin{theorem}\cite{BanchiCoutinhoGodsilSeverini} \label{thm:1}
Let $a$ and $b$ be vertices in a path with $n$ vertices. Then pretty good state transfer happens between $a$ and $b$ if and only if both conditions below hold.
\begin{enumerate}[(i)]
    \item $a + b = n+1$. In this case, for all $\theta_j \in \Theta_a$, define $\sigma_j = 0$ if $E_j \ket{a} = E_j \ket{b}$, and $\sigma_j = 1$ if $E_j \ket{a} = - E_j \ket{b}$.
    \item For any set of integers $\{\ell_j : \theta_j \in \Theta_a\}$ such that
    \[\sum_{\theta_j \in \Theta_a} \ell_j \theta_j = 0 \quad \text{and}\quad \sum_{\theta_j \in \Theta_a}\ell_j = 0,\]
    then
    \[\sum_{\theta_j \in \Theta_a} \ell_j \sigma_j \text{ is even}.\]    
\end{enumerate}
\end{theorem}

\section{New examples of pretty good state transfer}

The path graph $P_n$ is the graph on vertices $\{1,\ldots, n\}$ where vertex $a$ is adjacent to $a+1$ for all $a = 1,\ldots, n-1$.

\begin{theorem}\label{thm:main} Given any odd prime $p$ and positive integer $t$, there is pretty good state transfer in $P_{2^t p -1}$ between vertices  $a$ and $2^t p -a$, whenever $a$ is a multiple of $2^{t-1}$.
\end{theorem}
\begin{proof}
For simplicity, let $n = 2^t p -1$. For vertices $a$ and $(n+1-a)$, condition (i) of Theorem \ref{thm:1} is satisfied with $2 \sigma_j = 1 + (-1)^j$, by Lemma \ref{eq:3}.

The eigenvalues of the path $P_n$ belong to the cyclotomic field $\Q[\zeta_{2m}]$, where $m = n + 1$. More precisely,  
\[2 \cos\left( \frac{j \pi }{m} \right) = \zeta_{2m}^j + \zeta_{2m}^{-j}.\]
If $m = 2^k p$, then the cyclotomic polynomial is
\[\Phi_{2m}(x) = \sum_{i = 0}^{p-1} (-1)^i x^{2^k i}.\]

We will proceed by showing that part (ii) of Theorem \ref{thm:1} holds. If $a$ is a multiple of $2^{t-1}$, suppose there is a linear combination of the eigenvalues in $\Theta_{a}$, satisfying
\[\sum_{j=1}^{n} \ell_j \theta_j = 0,\]
where we make $\ell_j = 0$ if $\theta_j \notin \Theta_a$. Recall that the $a$th entry of the $\theta_j$-eigenvector is $\sin(a\pi j / (n+1))$. We see that $\theta_j$ belongs to $\Theta_a$ if and only if $2p$ does not divide $j$. 

We define the polynomial $P(x)$ as follows:
\[P(x) = \sum_{j=1}^{n} \ell_j x^j + \sum_{j=n+2}^{2n+1} \ell_{2n+2 - j} x^{j}\]
We see that $\zeta_{2m}$ is a root of $P(x)$ and, 
since $\Phi_{2m}(x)$ is the minimal polynomial of $\zeta_{2m}$, we see that $\Phi_{2m}(x)$ divides $P(x)$. 

Let $Q(x)$ be the following polynomial:
\begin{align*}
Q(x) &= \sum_{j = 1}^{2^t} \ell_j x^j + \sum_{j = 2^t + 1}^{2^t p-1} (\ell_j + \ell_{j - 2^t}) x^j - \ell_{2^t (p-1)} x^{2^t p} \\ &+ \sum_{j = 1}^{2^t -  1} (\ell_{2^t p - j} + \ell_{2^t (p - 1) + j} - \ell_j) x^{2^t p + j}.
\end{align*}

Consider $[x^k] \Phi_{2m}(x) Q(x)$. It is easy to see that $[x^k] \Phi_{2m}(x) Q(x) = [x^k] P(x)$ for $k = 0,\ldots, 2^t (p+1)-1$. Since the degree of $Q(x)$ is $2^t (p+1)-1$, we may conclude that $Q(x)$ is the unique polynomial of degree $2^t (p+1)-1$ such that 
\begin{equation}\label{eq:eq1}
    [x^k] \Phi_{2m}(x) Q(x) = [x^k] P(x)
\end{equation}
for $k = 0,\ldots, 2^t (p+1)-1$. In particular, the quotient $P(x)/ \Phi_{2m}(x)$ is a polynomial of degree $2^t (p+1)-1$ such that (\ref{eq:eq1}) holds, therefore
\[P(x) = \Phi_{2m}(x) Q(x).\]

From the coefficients of $x^k$ for $k > 2^t (p+1)-1$, it follows that, for $j = 2, 4, \ldots, 2^{t-1} - 2,$ and $i= 1,\ldots,(p - 1)/2$,
\begin{align*}
\ell_j - \ell_{2^t p - j} & = (-1)^i (\ell_{i 2^{t} \pm j} - \ell_{(p - i) 2^{t} \mp j}), \text{ and}\\
\ell_{i 2^{t-1}} - \ell_{(p-i)2^{t-1}} & = 0.
\end{align*}
Recall that $\ell_{2kp} = 0$ for any integer $k$.

Given $j \in \{2, 4, \ldots, 2^{t-1} - 2\}$, note that $j \neq 0 \pmod p$, and since $2^{t} \neq 0 \pmod p$, there is $i \in \Z_p$ such that $i 2^{t}\equiv j \pmod p$. If $1 \leq i \leq (p-1)/2$, then $\ell_{i 2^{t} - j} = \ell_{(p - i) 2^{t} + j} = 0$, and if $(p-1)/2 + 1 \leq i \leq p-1$, then $\ell_{i 2^{t} + j} = \ell_{(p - i) 2^{t} - j} = 0$. In either case, it follows that $\ell_j - \ell_{2^t p - j} = 0$.  Therefore $\ell_j = \ell_{2^t p - j}$ for all even $j$.

Thus, we see that 
\[\sum_{\theta_j \in \Theta_{a}} \ell_j \sigma_j = \sum_{j \text{ even}} \ell_r \equiv \ell_{2p} \equiv 0 \pmod 2,\]
which concludes the proof.
\end{proof}

\section{Open Problems}

Pretty good state transfer between end vertices of paths was classified in \cite{GodsilKirklandSeveriniSmithPGST}. In this paper, we have given an infinite family of paths where pretty good state transfer occurs between internal vertices, but not between the end vertices. A full classification of the orders $n$ where this occurs would be interesting and would complete the classification of pretty good state transfer in paths. 

In \cite{BanchiCoutinhoGodsilSeverini}, the authors study pretty good state transfer between the end vertices in a Heisenberg chain; they study the Hamiltonian $\exp(itL(P_n))$, where $L$ denotes the Laplacian adjacency matrix. In this case, they show that pretty good state transfer occurs between the end vertices of $P_n$ if and only if $n$ is a prime congruent to $1$ modulo $4$ or $n$ is a power of $2$. It is an open problem to find examples in which pretty good state transfer occurs between internal vertices while it does not occur between the end vertices. 

\section{Acknowledgements}
GC acknowledges the support of grants FAPESP 15/16339-2 and FAPESP 13/03447-6. CV acknowledges the support of an NSERC Canada Graduate Scholarship.  The three authors thank Chris Godsil for generously inviting Gabriel Coutinho to visit Waterloo, as well as for his general advice, support and for promoting a fruitful and positive research environment.

\end{document}